\documentclass[12pt]{amsart}

\usepackage{fullpage}
\usepackage{url}

\DeclareMathOperator{\RT}{RT}
\DeclareMathOperator{\ce}{ce}

\newtheorem{theorem}{Theorem}

\newtheorem{lemma}[theorem]{Lemma}

\newtheorem{observation}[theorem]{Observation}

\begin{document}

\title{The Lexicographically Least Binary Rich Word Achieving the Repetition Threshold}

\author{James Currie}\thanks{The work of James Currie is supported by the Natural Sciences and Engineering Research Council of Canada (NSERC), [funding reference number 2017-03901].}
\author{Narad Rampersad}\thanks{The work of Narad Rampersad is supported by the Natural Sciences and Engineering Research Council of Canada (NSERC), [funding reference number 2019-04111].}
\address{Department of Mathematics and Statistics \\
University of Winnipeg \\
515 Portage Avenue \\
Winnipeg, Manitoba R3B 2E9 (Canada)}
\email{\{j.currie,n.rampersad\}@uwinnipeg.ca}


\date{\today}

\maketitle

\section{Introduction}
A major branch of combinatorics on words studies words avoiding various powers or patterns. A typical question is whether there exists an infinite word over a certain alphabet avoiding a certain pattern. The earliest known result of this type is by Thue \cite{thue06}, who proved that there is an infinite word over a three-letter alphabet containing no factor of the form $h(xx)$ with $h$ a non-erasing morphism. 

We use $\Sigma_n$ to denote the $n$-letter alphabet $\Sigma_n=\{0,1,2,\ldots,n-1\}$. Let $p$ be an arbitrary finite string. Several generalizations of Thue's result have been explored.
\begin{enumerate}
\item  Does there exist some $n$ such that there is an infinite word over $\Sigma_n$ containing no factor of the form $h(p)$ with $h$ a non-erasing morphism?
\item For a fixed $n$, is there an infinite word over $\Sigma_n$ containing no factor of the form $h(p)$ with $h$ a non-erasing morphism?
\end{enumerate}
The first of these problems was shown to be decidable by Bean et.\ al. \cite{bean79} and independently by Zimin \cite{zimin80}. It is unknown whether the second problem is decidable.

A word of length $\ell$ and period $p$ is called a $k$-power, where $k=\ell/p$. A reformulation of Thue's result is that there is an infinite word over $\Sigma_3$ not containing a $2$-power. For integer $n\ge 2$, the {\em repetitive threshold} function is defined by 
$$\RT(n)=\sup\{k: \text{ every infinite word over $\Sigma_n$ contains a $k$-power}\}.$$

Thus Thue showed that $\RT(3)\le 2$. Dejean \cite{dejean72} showed that in fact $\RT(3)=7/4$, and conjectured that
\begin{displaymath}
\RT(n) = \begin{cases}
	  7/4,                 & \text{ if $n=3$;}\\
	  7/5,                 & \text{ if $n=4$;}\\
	  n/(n-1),                 & \text{ if $n\ne 3,4$.}\\
\end{cases}
\end{displaymath}
Dejean's conjecture was finally proved by Rao, and independently by Currie and Rampersad \cite{rao11,currie11}.
Words over an alphabet which realize the repetition threshold of the alphabet are called {\bf threshold words} and are extremal objects. In the case $n=2$, the threhold words are the binary overlap-free words, which have a large literature. (A good reference is the thesis of Rampersad \cite{rampersad07}.) With the solution of Dejean's conjecture, an indexed family of similar languages present themselves for study. As an example of such study, for threshold words on $\Sigma_n$ with $n\ge 27$, Currie et.\ al.\cite{currie20exp} have shown that the number of words grows exponentially with length.

Also branching off from the solution of Dejean's conjecture is the study of repetition thresholds for various classes of words. For example, various authors have found the repetition thresholds for binary rich words, for balanced sequences, and for circular words \cite{currie20,dolce23,gorbunova12}. Other types of repetition thresholds have also been studied, such as undirected repetition thresholds and Abelian repetition thresholds \cite{currie21,samsonov12}.

When investigating the existence of an infinite word over $\Sigma_n$ with some property, a natural approach is to generate and study long finite words with the property. Such words are typically generated by backtracking, and are therefore the lexicographically least words of a given length. Practically speaking then, solving avoidance problems often involves generating and parsing prefixes of the lexicographically least infinite word with a given property. Allouche et.\ al.\cite{allouche98} characterized the lexicographically least infinite overlap-free binary word starting with any specified prefix. Currie \cite{currie23} characterized the lexicographically least infinite {\em good} word, where the good words are closely related to the period-doubling morphism. However, the general study of lexicographically least infinite words with avoidance properties is in its infancy, and more examples are needed.

The current note combines the theme of repetition threshold with that of lexicographically least words. The 2020 paper of Currie et.\ al.\cite{currie20} established the repetition threshold for binary rich words. Studying such words by backtracking leads naturally to the question: What is the lexicographically least infinite binary rich word? We answer this question in this note.

\section{Preliminaries}

A {\em word} over {\em alphabet} $\Sigma_n$ is a finite or infinite sequence over $\Sigma_n$. We use lower case letters for finite words, and write, e.\ g., word $w=w_1w_2\cdots w_m$, where each $w_i\in \Sigma_n$. The {\em length} of $w$ is denoted by $|w|=m$. The word of length $0$ is called the {\em empty word}, and is denoted by $\epsilon$.  The {\em concatenation} of two words $u=u_1u_2\cdots u_s$ and $v=v_1v_2\cdots v_t$ is given by $uv=u_1u_2\cdots u_s v_1v_2\cdots v_t$. If $u, v, w, z$ are words and $w=uzv$, we call word $z$ a {\em factor} of $w$, word $u$ a {\em prefix} of $w$, and word $v$ a {\em suffix} of $w$. If $w=uv$, we define $u^{-1}w=v$.

A {\em morphism} from $\Sigma_n^*$ to $\Sigma_m^*$ is a function $f$ respecting concatenation; i.e., $f(xy)=f(x)f(y)$ for all $x,y\in \Sigma_n^*$. If $f^{-1}(\epsilon)=\{\epsilon\}$, we call $f$ {\em non-erasing}. 

We use bold-face letters for infinite words, writing ${\boldsymbol w}=w_1w_2w_3\cdots$, where each $w_i\in \Sigma_n$. The set of finite words over $\Sigma_n$ is denoted by $\Sigma_n^*$, and the set of infinite words is denoted by $\Sigma_n^\omega$.

Iteration of a morphism $f$ is written as exponentiation: 
\begin{displaymath}
f^i(x) = \begin{cases}
	  x,                 & \text{ if $i=0$;}\\
	  f(f^{i-1}(x)), & \text{ if $i>0$.}
\end{cases}
\end{displaymath}

If $f:\Sigma_n^*\rightarrow \Sigma_n^*$ is a morphism such that for some $a\in \Sigma_n$, $f(a)=au$,$u\ne \epsilon$, then $f^{n-1}(a)$ is a proper prefix of $f^n(a)$ for every positive integer $n$. We can then define ${\boldsymbol w}=\lim_{n\rightarrow\infty}f^n(a)$ to be the infinite word such that, for each $n$, word $f^n(a)$ is a prefix of ${\boldsymbol w}$.

Let $w$ be a finite word over $\Sigma_n$. Write  $w=w_1w_2\cdots w_m$ where each $w_i\in\Sigma_n$. The {\em  reversal} of $w$ is the word $w^R=w_mw_{m-1}\cdots w_1$. We call word $w$ a {\em palindrome} if $w=w^R$. Any word $w$ contains at most $|w|$ distinct palindromic factors. If $w$ in fact contains $|w|$ distinct palindromic factors, we say that $w$ is {\em rich}. A good reference on rich words is the paper of Glen et.\ al.\cite{glen09}. One of their results which we will use is

\begin{theorem}\cite[Theorem 2.14]{glen09}  For any finite or infinite word w, the following properties are equivalent:
\begin{itemize}
\item[i] $w$ is rich;
\item[ii] for any factor $u$ of $w$, if $u$ contains exactly two occurrences of a palindrome $p$ as a prefix and as a suffix
only, then $u$ is itself a palindrome.
\end{itemize}
\end{theorem}

A factor $u$ of $w$ containing exactly two occurrences of a factor $p$ as a prefix and as a suffix is called a {\em return word} of $p$.
An infinite word is defined to be rich if each of its finite factors is rich.

Let $w$ be a finite or infinite word. The {\em critical exponent} of $w$ is defined to be
$$\ce(w)=\sup\{k: \text{ $w$ contains a $k$-power}\}.$$
Let $L$ be a set of infinite words. The repetitive threshold of $L$ is defined to be
$$\RT(L)=\sup\{k: \text{ every word of $L$ contains a $k$-power}\}=\inf\{\ce(w):w\in L\}.$$
Thus $\RT(\Sigma_n^*)=\RT(\Sigma_n^\omega)=\RT(n)$.

Baranwal and Shallit~\cite{baranwal19} showed that there is an infinite binary rich word with critical exponent $2+\sqrt{2}/2$, and Currie et.\ al.\cite{currie20} proved that this word achieves the repetition threshold for infinite binary rich words.
Thus, if $L$ is the set of binary rich words, 
$\RT(L)=2+\sqrt{2}/2$.
Let ${\boldsymbol L}$ be the set of infinite binary rich words. The set ${\boldsymbol T}$ of {\em threshold words} is the set of infinite binary rich words whose critical exponent is the repetition threshold. Thus
$${\boldsymbol T}=\{{\boldsymbol w}\in{\boldsymbol L}:\ce({\boldsymbol w})=2+\sqrt{2}/2\}.$$

Define morphisms $f:\Sigma_3^*\rightarrow\Sigma_2^*$ and $g,h:\Sigma_3^*\rightarrow\Sigma_3^*$ by
\begin{align*}
    f(\tt{0})&=\tt{0}\\
    f(\tt{1})&=\tt{01}\\
    f(\tt{2})&=\tt{011}\\[6pt]
    g(\tt{0})&=\tt{011}\\
    g(\tt{1})&=\tt{0121}\\
    g(\tt{2})&=\tt{012121}\\[6pt]
    h(\tt{0})&=\tt{01}\\
    h(\tt{1})&=\tt{02}\\
    h(\tt{2})&=\tt{022}
\end{align*}

Word $f(h^\omega(\tt{0}))$ is the word constructed by Baranwal and Shallit~\cite{baranwal19}.  
The word $f(g(h^\omega(\tt{0})))$ was shown to be a binary rich word with the same critical exponent by Currie et.\ al.\ \cite{currie20}.

The lexicographic order on $\Sigma_n^*$ and $\Sigma_n^\omega$ is defined as follows: \begin{itemize}
\item We order letters in the natural way: $0<1<2<\cdots<n-1.$ We also insist that $\epsilon<0$. 
\item Let the longest common prefix of $u$ and $v$ be $p$. We say that $u<v$ if and only if the first letter of $p^{-1}u$ is less than the first letter of $p^{-1}v$, where the first letter of $\epsilon$ is taken to be $\epsilon$.
\end{itemize}

One checks that morphisms $f$, $g$, and $h$ are order-preserving: Let $\phi\in\{f,g,h\}$. If $u\le v$ then $\phi(u)\le\phi(v)$.

\begin{theorem}[Main Theorem]
Word $\ell=f(01g(h^\omega(\tt{0})))$ is the lexicographically least word in ${\boldsymbol T}$.
\end{theorem}  

\section{Proof of Main Theorem}

We say that a word $w\in\Sigma_2^*$ is {\bf good} if it is both rich and $14/5$-free. We say that an infinite word over $\Sigma_2$ is good if its factors are good.

\begin{observation}\label{f}
Since an infinite good word ${\boldsymbol w}$ must be $3$-free, it can be written as $pf({\boldsymbol u})$, where $p\in\{\epsilon,1,11\}$, and ${\boldsymbol u}\in\Sigma_3^\omega$.
\end{observation}

We use slight amplifications of the results of Currie et.\ al.\ \cite{currie20}:

\begin{lemma}~\label{outer} Suppose $f({\boldsymbol u})$ is good, where ${\boldsymbol u}\in011\Sigma_3^\omega$. Then \begin{enumerate}
\item Word ${\boldsymbol u}=g({\boldsymbol W})$ for some word ${\boldsymbol W}\in\Sigma_3^\omega$.
\item Word ${\boldsymbol W}$ has the form $h({\boldsymbol U})$ for some word ${\boldsymbol U}\in\Sigma_3^\omega$.
\end{enumerate}
\end{lemma}

Lemma~\ref{outer} follows from the proof of Lemma~9 of Currie et.\ al.\ \cite{currie20}.

\begin{lemma}\label{inner}
Let ${\boldsymbol u}\in0\Sigma_3^\omega$. Suppose that for some positive integer $n$, one of $f(g(h^n({\boldsymbol u})))$ and $f(h^n({\boldsymbol u}))$ is good. Then ${\boldsymbol u}=h({\boldsymbol W})$ for some word ${\boldsymbol W}\in\Sigma_3^\omega$.
\end{lemma}

Lemma~\ref{inner} follows from the proofs of Lemmas~10 and 11 of Currie et.\ al.\ \cite{currie20}.

\begin{theorem}\label{good} The infinite binary word ${\boldsymbol v}=f(g(h^\omega(0)))$ is good.
\end{theorem}

Theorem~\ref{good} follows from Theorems~15 and 17 of Currie et.\ al.\ \cite{currie20}.

We begin with a preliminary lemma. 
\begin{lemma}\label{001010}
The lexicographically least infinite good word with prefix $001010$ is
$${\boldsymbol v}=f(g(h^\omega(0))).$$ 
\end{lemma}
\begin{proof}
Suppose that ${\boldsymbol V}$ is an infinite good word with prefix $001010$, and ${\boldsymbol V}\le{\boldsymbol v}$. By Observation~\ref{f}, write ${\boldsymbol V}=f({\boldsymbol u})$, where ${\boldsymbol u}\in\Sigma_3^\omega$. Since ${\boldsymbol V}$ has prefix $001010$, word ${\boldsymbol u}$ has prefix $011$. It follows from Lemma~\ref{outer} that 
${\boldsymbol u}=g({\boldsymbol U})$ for some word ${\boldsymbol U}\in\Sigma_3^\omega$, where ${\boldsymbol U}$ has the form $h({\boldsymbol W}_1)$ for some word ${\boldsymbol W}_1\in\Sigma_3^\omega$.

Since $f$ is order-preserving, ${\boldsymbol u}\le g(h^\omega(0))$. Since $g$ is order-preserving, ${\boldsymbol U}\le h^\omega(0)$. Since $h$ is order-preserving, ${\boldsymbol W}_1\le h^\omega(0)$. In particular, since the first letter of $h^\omega(0)$ is $0$, the first letter of ${\boldsymbol W}_1$ is $0$. Using Lemma~\ref{inner}, write ${\boldsymbol W}_1=h({\boldsymbol W}_2)$ for some word ${\boldsymbol W}_2\in\Sigma_3^\omega$. Again, since $h$ is order-preserving, the first letter of ${\boldsymbol W}_2$ is $0$. By induction, we find that for each positive integer $n$ we have ${\boldsymbol W}_1=h^{n-1}({\boldsymbol W}_n)$, for some word ${\boldsymbol W}_n\in0\Sigma_3^\omega$. It follows that $h^n(0)$ is a prefix of ${\boldsymbol W}_1$ for each $n$, so that ${\boldsymbol W}_1=h^\omega(0)$.

We conclude that the lexicographically least infinite good word with prefix $001010$ is
$${\boldsymbol v}=f(g(h^\omega(0))).$$ 
\end{proof}
The Main Theorem follows from the following three lemmas.
\begin{lemma}\label{lexleast} Let ${\boldsymbol m}$ be an infinite good word. Let 
$${\boldsymbol \ell}=f(01g(h^\omega(0))).$$ Then 
$${\boldsymbol \ell}\le {\boldsymbol m}.$$
\end{lemma}
\begin{proof}
The least binary $3$-free word of length 8 is $00100100$. However, $00100100$ cannot be extended on the right to a binary $3$-free word. It follows that  $001001010\le{\boldsymbol m}.$ If 
${\boldsymbol m}\le {\boldsymbol \ell}$, then
$$f(01g(0))0=001001010\le {\boldsymbol m}\le {\boldsymbol \ell}=f(01g(h^\omega(0))),$$ forcing ${\boldsymbol m}$ to have prefix $001001010$. Then $(001)^{-1}{\boldsymbol m}$ is an infinite good word with prefix $001010$. By Lemma~\ref{001010}, this forces
$$f(g(h^\omega(0)))\le (001)^{-1}{\boldsymbol m}$$
forcing
$${\boldsymbol \ell}=001f(g(h^\omega(0)))\le {\boldsymbol m}.$$
\end{proof}
\begin{lemma}\label{recurrent} Word $01f(g(h^\omega(0)))$ is recurrent.
\end{lemma}
\begin{proof} Word $g(h^\omega(0))$ is recurrent. However, the only letter preceding a $0$ in $g(h^\omega(0))$ is $1$, so that if $p$ is a prefix of $g(h^\omega(0))$, word $1p$ must be a (necessarily recurrent) factor of $g(h^\omega(0))$. Any $u$ factor of $01f(g(h^\omega(0)))$ is a factor of $01f(p)=f(1p)$ for some prefix $p$ of $g(h^\omega(0))$. Since $1p$ is recurrent in $g(h^\omega(0))$, $f(1p)$ is recurrent in $f(g(h^\omega(0)))$, and so is $u$. Then $u$ is recurrent in $01f(g(h^\omega(0)))$. we conclude that $01f(g(h^\omega(0)))$ is recurrent.
\end{proof}

\begin{lemma} The word
$${\boldsymbol \ell}=f(01g(h^\omega(0)))$$  is $14/5$-free.
\end{lemma}
\begin{proof} Currie et.\ al.\cite{currie20} proved that $f(g(h^\omega(0)))$  is $14/5$-free. By Lemma~\ref{recurrent}, $01f(g(h^\omega(0)))$ has the same factors as $f(g(h^\omega(0)))$  and is also $14/5$-free. Therefore, any $14/5^+$ power in $f(g(h^\omega(0)))=001f(g(h^\omega(0)))$ must be a prefix. The word $00100$ is a prefix of $f(01g(h^\omega(0)))$, but does not occur in
$f(g(h^\omega(0)))$. It follows that any $14/5^+$ power which is a prefix of $f(g(h^\omega(0)))$ has period 4 or less. A very short finite check shows no such $14/5^+$ power is a prefix of $f(01g(h^\omega(0)))$.

\end{proof}
\begin{lemma} The word
$${\boldsymbol \ell}=f(01g(h^\omega(0)))$$  is rich.
\end{lemma}
\begin{proof} Currie et.\ al.\cite{currie20} proved that $f(g(h^\omega(0)))$  is rich. By Lemma~\ref{recurrent}, $01f(g(h^\omega(0)))$ has the same factors as $f(g(h^\omega(0)))$  and is also rich. 

Suppose that $001f(g(h^\omega(0)))$ is not rich. It will therefore have a complete return which is not a palindrome. Since $01f(g(h^\omega(0)))$ is rich, some prefix of $001f(g(h^\omega(0)))$ must be a complete return to a palindrome which is not a palindrome. Let this prefix be $pqp$ where $p$ is a palindrome and $q$ is not. The palindrome $00100$ is a prefix of $f(01g(h^\omega(0)))$, but does not occur in
$f(g(h^\omega(0)))$. It follows that $|p|\le 4$. The only possibility is seen to be $p=00$. However the complete return to $00$ is $00100$, which is a palindrome. This is a contradiction.
\end{proof}

\end{document}